\documentclass[10pt]{article}   
\usepackage{latexsym}     
\usepackage{enumerate}
\usepackage{amsthm}     
\usepackage{amssymb}     
\usepackage{amsmath}     
\usepackage{proof}     
\usepackage{graphics}     
\usepackage{color}     
\usepackage[all]{xy}    
\usepackage[comma,authoryear]{natbib}
\newtheorem{theorem}{Theorem}[section]
\newtheorem{lemma}[theorem]{Lemma}
\newtheorem{claim}[theorem]{Claim}
\newtheorem{corollary}[theorem]{Corollary}

\theoremstyle{definition}

 \newcommand{\cS}{\ensuremath{\mathcal{S}}}  
 \newcommand{\cSd}{\ensuremath{\mathcal{S}^\dagger}} 
 \newcommand{\cN}{\ensuremath{\mathcal{N}}}  
 \newcommand{\cNd}{\ensuremath{\mathcal{N}^\dagger}} 
 \newcommand{\cT}{\ensuremath{\mathcal{T}}}  
 \newcommand{\leZ}{\ensuremath{\exists_{\leq 0}}} 
 \newcommand{\mZ}{\ensuremath{\exists_{> 0}}}    
 \newcommand{\leO}{\ensuremath{\exists_{\leq 1}}} 
 \newcommand{\mO}{\ensuremath{\exists_{> 1}}}    
 \newcommand{\exO}{\ensuremath{\exists_{= 1}}}   
 \newcommand{\lei}{\ensuremath{\exists_{\leq i}}} 
 \newcommand{\mi}{\ensuremath{\exists_{> i}}}    

 \newcommand{\bbP}{\ensuremath{\mathbb{P}}}  

 \newcommand{\bP}{\ensuremath{\mathbf{P}}}   

 \newcommand{\cL}{\ensuremath{\mathcal{L}}}  

 \newcommand{\fA}{\ensuremath{\mathfrak{A}}} 
 \newcommand{\fB}{\ensuremath{\mathfrak{B}}} 

 \newcommand{\sX}{{\sf X}}                   

 \newcommand{\set}[1]{\{#1\}}
 \renewcommand{\phi}{\varphi}

 \newcommand{\NLOGSPACE}{\textsc{NLogSpace}}

 \newcommand{\NPTIME}{\textsc{NPTime}}
\begin{document}

 \title{The Syllogistic with Unity}

 \author{Ian Pratt-Hartmann\\
School of Computer Science\\
University of Manchester\\
Manchester M13 9PL, UK.\\
\tt{ipratt@cs.man.ac.uk}}

\date{}

\maketitle

\begin{abstract}
\noindent
We extend the language of the classical syllogisms with the
sentence-forms ``At most 1 $p$ is a $q$'' and ``More than 1 $p$ is a
$q$''.  We show that the resulting logic does not admit a finite set
of syllogism-like rules whose associated derivation relation is sound
and complete, even when {\em reductio ad absurdum} is allowed.
\end{abstract}
\section{Introduction}
\label{sec:intro} 
By the {\em classical syllogistic},
we understand the set of English sentences of the forms
\begin{equation}
\begin{array}{lll}
\mbox{No $p$ is a $q$} 
&  & 
\mbox{Some $p$ is a $q$}\\
\mbox{Every $p$ is a $q$} & & 
\mbox{Some $p$ is not a $q$},
\end{array}
\label{eq:EnglishS}
\end{equation}
where $p$ and $q$ are common (count) nouns.  By the {\em extended
  classical syllogistic}, we understand the classical syllogistic
together with the set of quasi-English sentences of the forms
\begin{equation}
\begin{array}{lll}
\mbox{Every non-$p$ is a $q$} 
&  & 
\mbox{Some non-$p$ is not a $q$}.
\end{array}
\label{eq:EnglishSdagger}
\end{equation}
It is known that there exists a sound and complete proof system for
the classical syllogistic in the form of a finite set of
syllogism-like proof-rules~\citep{s1:smiley73,s1:corcoran}. Such a
proof system also exists for the extended classical syllogistic;
moreover, in both cases, {\em reductio ad absurdum}---in other words,
the strategy of indirect proof---can be dispensed
with~\citep{s1:p-h+m09}. The satisfiability problem for either of
these languages is easily seen to be \NLOGSPACE-complete, by a routine
reduction to (and from) the problem 2-SAT.

Both the classical syllogistic and its extended variant may be
equivalently reformulated using the numerical quantifiers ``At most 0
\ldots'' and ``More than 0 \ldots''. The forms of the classical
syllogistic thus become, respectively
\begin{equation}
\begin{array}{lll}
\mbox{At most 0 $p$s are $q$s} 
&  & 
\mbox{More than 0 $p$s are $q$s}\\
\mbox{At most 0 $p$s are not $q$s} & & 
\mbox{More than 0 $p$s are not $q$s},
\end{array}
\label{eq:EnglishS0}
\end{equation}
while the additional forms of the extended classical syllogistic
become
\begin{equation}
\begin{array}{lll}
\mbox{At most 0 non-$p$s are not $q$s}
&  &
\mbox{More than 0 non-$p$s are not $q$s}.
\end{array}
\label{eq:EnglishSdagger0}
\end{equation}
These reformulations invite generalization.  By the {\em numerical
  syllogistic}, we understand the set of English sentences of the
forms
\begin{equation}
\begin{array}{lll}
\mbox{At most $i$ $p$s are $q$s \hspace{0.8cm}} 
&  & 
\mbox{More than $i$ $p$s are $q$s}\\
\mbox{At most $i$ $p$s are not $q$s} & & 
\mbox{More than $i$ $p$s are not $q$s},
\end{array}
\label{eq:EnglishSN}
\end{equation}
where $p$ and $q$ are common count nouns and $i$ is a (decimal
representation of a) non-negative integer.  By the {\em extended
  numerical syllogistic}, we understand the numerical syllogistic
together with the set of quasi-English sentences of the forms
\begin{equation}
\begin{array}{lll}
\mbox{At most $i$ non-$p$s are not $q$s} 
&  &
\mbox{More than $i$ non-$p$s are not $q$s}.
\end{array}
\label{eq:EnglishSdaggerN}
\end{equation}
In other words, the classical syllogistic is simply the fragment of
the numerical syllogistic in which all numbers are bounded by 0; and
similarly for the extended variants.  The first systematic
investigation of the numerical syllogistic known to the author is that
of~\citet[Ch.~VIII]{s1:deMorgan47}, though this work was closely
followed by treatments in~\cite{s1:boole1868}~\citep[reprinted
  as][Sec.~IV]{s1:booleReprint52},
and~\cite{s1:jevons1871}~\citep[reprinted as][Part~I,
  Sec.~IV]{s1:jevonsReprint1890}.  For a historical overview of this
episode in logic, see~\cite{s1:gg00}.  De Morgan presented a list of
what he took to be the valid numerical syllogisms; and latter-day
systems may be found in~\citet{s1:hp67} and~\citet{s1:murphree98}.  It
can be shown, however, that there exists {\em no} sound and complete
syllogism-like proof system for the numerical syllogistic, even in the
presence of {\em reductio ad absurdum}; and similarly for the extended
numerical syllogistic~\citep{s1:ph09}.  In addition, the
satisfiability problems for the numerical syllogistic and the extended
numerical syllogistic are both \NPTIME-complete~\citep{s1:ph08}.

Thus, the numerical syllogistic differs from the classical syllogistic
in its proof-theoretic and complexity-theoretic properties.  The
purpose of the present paper is to locate the source of this
difference more precisely. Specifically, we consider the {\em
  syllogistic with unity}, which we take to consist of the classical
syllogistic together with the forms
\begin{equation}
\begin{array}{lll}
\mbox{At most 1 $p$ is a $q$} 
&  & 
\mbox{More than 1 $p$ is a $q$}\\
\mbox{At most 1 $p$ is not a $q$} & & 
\mbox{More than 1 $p$ is not a $q$},
\end{array}
\label{eq:EnglishS1}
\end{equation}
along with its extended variant, which additionally features the forms
\begin{equation}
\begin{array}{lll}
\mbox{At most 1 non-$p$ is not a $q$} 
&  & 
\mbox{More than 1 non-$p$ is not a $q$}.
\end{array}
\label{eq:EnglishSdagger1}
\end{equation}
In other words, the syllogistic with unity is simply the fragment of
the numerical syllogistic in which all numbers are bounded by 1; and
similarly for the extended variants.  The syllogistic with unity gives
rise to argument patterns having no counterparts in the classical
syllogistic. For example,
\begin{equation}
\begin{array}{l}
\mbox{At most 1 $o$ is a $p$}\\
\mbox{At most 1 $o$ is not a $p$}\\
\mbox{At most 1 $q$ is not an $o$}\\
\mbox{\underline{More than 1 $q$ is not an $r$}}\\
\mbox{At most 1 $q$ is an $r$}
\end{array}
\label{eq:argE}
\end{equation}
is evidently a valid argument. For the first two premises ensure that
there are at most two $o$s, whence the third premise ensures that
there are at most three $q$s; but the fourth premise states that at
least two of these are not $r$s.

Syntactically speaking, the syllogistic with unity lies closer to the
classical syllogistic than it does to the numerical syllogistic,
because, like the former, but unlike the latter, it features only
finitely many logical forms.  This fact notwithstanding, we show in
the sequel that there exists no sound and complete syllogism-like
proof system for the syllogistic with unity, even in the presence of
{\em reductio ad absurdum}; and similarly for its extended
variant. We also observe that the satisfiability problem for
either of these languages remains \NPTIME-complete. Thus, the smallest
conceivable extension of the classical syllogistic by means of
additional counting quantifiers yields the proof-theoretic and
complexity-theoretic properties of the entire numerical syllogistic.
Generalizing this result, we consider the family of languages obtained
by restricting the numerical syllogistic so that all numbers are
bounded by $z$, where $z$ is any positive integer; and similarly for
the extended numerical syllogistic. We show that, for all these
languages, there exists no sound and complete syllogism-like proof
system, even in the presence of {\em reductio ad absurdum}.

The syllogistic with unity exhibits some similarities with the
intriguing logical system proposed by~\citet[pp.~249--317]{s1:hamilton60}.
According to Hamilton, the predicates of traditional syllogistic
sentence-forms contain implicit existential quantifiers, so that, for
example, ``All $p$ is $q$'' is to be understood as ``All $p$ is some
$q$''. Further, these implicit existential quantifiers can be
meaningfully dualized to yield novel sentence-forms, thus: ``All $p$
is all $q$.''~\citep[A similar language was actually proposed in the
  earlier, but lesser-known][]{s1:bentham27}. Hamilton's account of
the meanings of these sentences is, it must be said, unclear. However,
a natural interpretation is obtained by taking the copula simply to
denote the relation of identity.  Thus, for example, the sentence
``All $p$ is all $q$'' is formalized by $\forall x(p(x) \rightarrow
\forall y(q(y) \rightarrow x = y))$---equivalently, either there are
no $p$s, or there are no $q$s, or there is exactly 1 $p$ and exactly 1
$q$, and they are identical~\citep[see, e.g.][] {s1:fogelin76b}.
Under this interpretation, the pair of Hamiltonian sentences ``All $p$
are all $p$'' and ``Some $p$ are some $p$'' then states that there
exists exactly one $p$---something assertable in the syllogistic with
unity.  In general, however, the two languages are expressively
incomparable; in particular, the Hamiltonian syllogistic provides no
means of stating that exactly one $p$ is a $q$ (with $p$ and $q$
different).  Moreover, they exhibit different proof-theoretic
properties: unlike the syllogistic with unity, Hamilton's language
does indeed have a sound and complete syllogistic proof-system, though
some form of indirect proof is essential~\citep{s1:pratt-hartmann11}.

\section{Syntax and  semantics}
\label{sec:synsem}
Fix a countably infinite set $\bP$.  We refer to any element of $\bP$
as an {\em atom}.  A {\em literal} is an expression of either of the
forms $p$ or $\bar{p}$, where $p$ is an atom.  A literal which is an
atom is called {\em positive}, otherwise, {\em negative}. If $\ell =
\bar{p}$ is a negative literal, then we denote by $\bar{\ell}$ the
positive literal $p$. If $z$ is a non-negative integer, an
$\cS_z$-{\em formula} is any expression of the forms
\begin{equation}
\begin{array}{lll}
  \lei(p,\ell) & \hspace{1cm} &   \mi(p,\ell),
\end{array}
\label{eq:syntaxS1}
\end{equation}
where $p$ is an atom, $\ell$ is a literal and $0 \leq i \leq z$.  An
$\cSd_z$-{\em formula} is any expression of the forms
\begin{equation}
\begin{array}{lll}
  \lei(\ell,m) & \hspace{1cm} &   \mi(\ell,m),
\end{array}
\label{eq:syntaxSd1}
\end{equation}
where $\ell$ and $m$ are literals and $0 \leq i \leq z$.  We denote
the set of $\cS_z$-formulas simply by $\cS_z$, and similarly for
$\cSd_z$. Where the language is clear from context, we speak simply of
{\em formulas}.  Evidently: $\cS_z \subseteq \cSd_z$, $\cS_z \subseteq
\cS_{z+1}$, and $\cSd_z \subseteq \cSd_{z+1}$.  We denote the union of
all the languages $\cS_z$ by $\cN$, and the union of all the languages
$\cSd_z$ by $\cNd$.

A {\em structure} is a pair $\fA = \langle A, \{ p^\fA \}_{p \in \bP}
\rangle$, where $A$ is a non-empty set, and $p^\fA \subseteq A$, for
every $p \in \bP$. The set $A$ is called the {\em domain} of $\fA$. We
extend the map $p \mapsto p^\fA$ to negative literals by setting, for
any atom $p$,
\begin{equation*}
\bar{p}^\fA = A \setminus p^\fA.
\end{equation*}
Intuitively, we may think of the elements of $\bP$ as common
count-nouns, such as ``pacifist'', ``quaker'', ``republican'',
etc., and if $a \in \ell^\fA$, we say that $a$ {\em satisfies} $\ell$
{\em in} $\fA$, and regard $a$ as having the property denoted by
$\ell$. Thus, we may gloss any negative literal $\bar{p}$ as
``non-$p$'' or ``not a $p$'' depending on grammatical context.  If
$\fA$ is a structure, we write $\fA \models \lei(\ell,m)$ if
$|\ell^\fA \cap m^\fA| \leq i$, and $\fA \models \mi(\ell,m)$ if
$|\ell^\fA \cap m^\fA| > i$. If $\fA \models \phi$, we say that $\phi$
is {\em true} in the structure $\fA$.  Thus, we may gloss
$\lei(\ell,m)$ as ``At most $i$ $\ell$s are $m$s'', and $\mi(\ell,m)$ as
```More than $i$ $\ell$s are $m$''.  If $i >0$, we write
$\exists_{=i}(\ell, m)$ as an abbreviation for the pair of formulas
$\set{\exists_{> (i-1)}(\ell,m), \lei(\ell,m)}$.  Where no confusion
results, we occasionally treat this pair as a single formula, which we
may gloss as ``Exactly $i$ $\ell$s are $m$s.''

Evidently, the languages $\cS_0$, $\cS_1$ and $\cN$ formalize the
classical syllogistic, the syllogistic with unity, and the numerical
syllogistic, respectively; similarly, $\cSd_0$, $\cSd_1$ and $\cNd$
formalize their respective extended variants. Observe that the above
semantics render formulas symmetric in their arguments: for example,
$\fA \models \lei(\ell,m)$ if and only if $\fA \models \lei(m,\ell)$,
and similarly for formulas featuring the quantifiers
$\mi$. Accordingly, we shall henceforth regard these arguments as
unordered: that is, we identify the formulas $\lei(\ell,m)$ and
$\lei(m,\ell)$, and similarly for $\mi$. This will help to reduce
notational clutter in some of the proofs.

If $\Theta$ is a set of formulas, we write $\fA \models \Theta$ if,
for all $\theta \in \Theta$, $\fA \models \theta$.  A formula $\theta$
is {\em satisfiable} if there exists a structure $\fA$ such that $\fA
\models \theta$; a set of formulas $\Theta$ is {\em satisfiable} if
there exists $\fA$ such that $\fA \models \Theta$.  We call a formula
of the form $\mi(p,\bar{p})$ an {\em absurdity}, and use $\bot$
to denote, indifferently, any absurdity. Evidently, $\bot$ is
unsatisfiable.  If, for all structures $\fA$, $\fA \models \Theta$
implies $\fA \models \psi$, we say that $\Theta$ {\em entails} $\psi$,
and write $\Theta \models \psi$. In the case where $\Theta =
\set{\theta}$, we say that $\theta$ {\em entails} $\psi$. Thus, for
example, the entailment
\begin{equation}
\set{\leO(o, p), \leO(o, \bar{p}), \leO(q, \bar{o}), \mO(q, \bar{r})} \models
  \leO(q, r)
\label{eq:argS1}.
\end{equation}
formalizes the valid argument~\eqref{eq:argE}.

If $\phi = \lei(\ell,m)$, we write $\bar{\phi}$ to denote
$\mi(\ell,m)$; and if $\phi = \mi(\ell,m)$, we write $\bar{\phi}$ to
denote $\lei(\ell,m)$. Thus, $\bar{\bar{\phi}} = \phi$, and, in any
structure $\fA$, $\fA \models \phi$ if and only if $\fA \not \models
\bar{\phi}$. Informally, we may regard $\bar{\phi}$ as the {\em
  negation} of $\phi$. It will sometimes be convenient to restrict
attention to formulas featuring only a limited selection of atoms.  If
$\bP' \subseteq \bP$, and $\cL$ is any of the languages
$\cS_z$ or $\cSd_z$, we denote the set of $\cL$-formulas $\phi$
involving only atoms in $\bP'$ by $\cL(\bP')$. Since $\phi \in
\cL(\bP')$ evidently implies $\bar{\phi} \in \cL(\bP')$, we may regard
all these languages as {\em closed under negation}. We call a subset
$\Phi \subseteq \cL(\bP')$ {\em complete for} $\cL(\bP')$ if, for
every $\phi \in \cL(\bP')$, either $\phi \in \Phi$ or $\bar{\phi} \in
\Phi$; reference to $\cL(\bP')$ is suppressed if clear from context.

Complete sets of formulas will play an important role in the sequel,
and we employ the following abbreviations to help define them.  Where
the language ($\cS_z$ or $\cSd_z$) is clear from context, and $0 \leq
i \leq z$, we write $\exists^*_{\leq i}(\ell, m)$ for the set of
formulas
\begin{equation*}
\set{\exists_{\leq i}(\ell, m), \dots, \exists_{\leq z}(\ell, m)},
\end{equation*}
and $\exists^*_{> i}(\ell, m)$ for the set of formulas
\begin{equation*}
\set{\exists_{> 0}(\ell, m), \dots, \exists_{> i}(\ell, m)}.
\end{equation*}
In addition, for $0 < i \leq z$, we write $\exists^*_{= i}(\ell, m)$
for the set of formulas
\begin{equation*}
\set{\exists_{> 0}(\ell, m), \dots, \exists_{> i-1}(\ell, m), 
\exists_{\leq i}(\ell, m), \dots, \exists_{\leq z}(\ell, m)}.
\end{equation*}
(Thus, in the languages $\cS_1$ and $\cSd_1$, $\exists^*_{=
  1}(\ell, m)$ and $\exists_{= 1}(\ell, m)$ coincide.)  It is easy
to see that, for any literals $\ell$, $m$, any structure $\fA$, and
any $i$ ($0 \leq i \leq z$), $\fA \models \exists^*_{\leq i}(\ell,m)$
if and only if $\fA \models \exists_{\leq i}(\ell,m)$; similarly, $\fA
\models \exists^*_{> i}(\ell,m)$ if and only if $\fA \models
\exists_{> i}(\ell,m)$. In addition, for $0 < i \leq z$, $\fA \models
\exists^*_{= i}(\ell,m)$ if and only if $\fA \models \exists_{=
  i}(\ell,m)$; moreover, exactly one of $\fA \models \exists^*_{\leq
  (i-1)}(\ell,m)$, $\fA \models \exists^*_{= i}(\ell,m)$ or $\fA
\models \exists^*_{> i}(\ell,m)$ holds.

As mentioned above, the satisfiability problems for $\cS_0$ and
$\cSd_0$---i.e.~the classical syllogistic and the extended classical
syllogistic---are both $\NLOGSPACE$-complete. We end this section with
a contrasting result on the complexity of satisfiability for $\cS_z$
and $\cSd_z$, where $z > 0$.
\begin{theorem}
For all $z > 0$, the problem of determining the satisfiability of a
given set of $\cS_z$-formulas is \NPTIME-complete, and similarly for
$\cSd_z$-formulas.
\label{theo:NP}
\end{theorem}
\begin{proof}
Let $\Phi$ be any satisfiable set of $\cSd_z$-formulas; we claim that
$\Phi$ has a model over a domain of size at most $(z+1)|\Phi|$.  Indeed,
suppose $\fA \models \Phi$.  For any formula $\mi(\ell,m)$, select
$(i+1)$ elements satisfying $\ell$ and $m$. Let $B$ be the set of
selected elements, and let $\fB$ be the restriction of $\fA$ to
$B$. It is obvious that $\fB \models \Phi$, and $|B| \leq
(z+1)|\Phi|$, proving the claim.  Membership of the satisfiability
problem for $\cSd_z$ in \NPTIME{} follows. 

It remains to show \NPTIME-hardness of the satisfiability problem for
$\cS_1$.  Let $\cT$ be the language consisting of $\cS_1$ together
with formulas of the forms $\exists_{\leq 3}(p,p)$, where $p$ is an
atom. It is shown in~\citet[Lemma~1]{s1:ph08} that the satisfiability
problem for $\cT$ is \NPTIME-hard, using a straightforward reduction
of graph-3-colourability. We need only reduce the satisfiability
problem for $\cT$ to that for $\cS_1$.  Let $\Phi$ be any set of
$\cT$-formulas, then. For any formula $\phi = \exists_{\leq 3}(p,p)$,
let $o$, $o'$ be new atoms, and replace $\phi$ by the set of
$\cS_1$-formulas $\set{\leO(p,\bar{o}), \quad \leO(o,o'), \quad
  \leO(o,\bar{o}')}$.  Let the resulting set of $\cS_1$-formulas be
$\Psi$.  Evidently, $\Psi$ entails every formula of $\Phi$;
conversely, any structure $\fA$ such that $\fA \models \Phi$ can
easily be expanded to a structure $\fA'$ such that $\fA \models
\Psi$. This completes the reduction.
\end{proof}
We remark that, when considering the satisfiability problem for
$\cSd_z$, the integer $z$ is a constant: thus, there is only a fixed
number of quantifiers $\exists_{\leq i}$ or $\exists_{> i}$, so that
we do not need to worry about the coding scheme (unary or binary) for
the numerical subscripts.  By contrast, for the language $\cNd$, which
features all these quantifiers, the coding of numerical subscripts is
a significant issue.  That the satisfiability problem for $\cNd$
remains in $\NPTIME$---even when numerical subscripts are coded as
bit-strings---requires a clever combinatorial argument due
to~\citet{s1:eisenbrand+shmonin06}; for details, see~\citet{s1:ph08}.

\section{Syllogistic proof systems}
\label{sec:proof}
Let $\cL$ be any of the languages $\cS_z$ or $\cSd_z$ ($z \geq 0$).
A {\em syllogistic rule} in $\cL$ is a pair $\Theta/\theta$, where
$\Theta$ is a finite set (possibly empty) of $\cL$-formulas, and
$\theta$ an $\cL$-formula.  We call $\Theta$ the {\em antecedents} of
the rule, and $\theta$ its {\em consequent}.  We generally display
rules in `natural-deduction' style. For example,
\begin{equation}
\begin{array}{ll}
\infer{\mZ(p,o)}
       {\leZ(q,\bar{o}) & 
             \qquad \mZ(p,q)}    
\hspace{0.5cm} & 
\infer[,]{\mZ(p,\bar{o})}
                   {\leZ(q, o) & 
             \qquad \mZ(p,q)}    
\end{array}
\label{eq:dariiFerio}
\end{equation}
where $p$, $q$ and $o$ are atoms, are syllogistic rules in $\cS_0$
(hence in all larger languages); they correspond to the traditional
syllogisms {\em Darii} and {\em Ferio}, respectively:
 \begin{center}
 \begin{minipage}{5cm}
 \begin{tabbing}
 Every $q$ is an $o$\\
 \underline{Some $p$ is a $q$}\\
 Some $p$ is an $o$
 \end{tabbing}
 \end{minipage}
 \hspace{0.5cm}
 \begin{minipage}{5cm}
 \begin{tabbing}
 No $q$ is an $o$\\
 \underline{Some $p$ is a $q$}\\
 Some $p$ is not an $o$.
 \end{tabbing}
 \end{minipage}
 \end{center}
We call a syllogistic rule {\em sound} if its antecedents entail its
consequent. Thus, the syllogistic rules~\eqref{eq:dariiFerio} are sound. More
generally, for all $i$, $j$, $z$ ($0 \leq i \leq j \leq z$),
\begin{equation*}
\begin{array}{ll}
\infer{\exists_{>(j-i)} (p,o)}
       {\exists_{\leq i}(q,\bar{o}) & 
             \qquad \exists_{>j} (p,q)}    
\hspace{0.5cm} & 
\infer{\exists_{>(j-i)} (p,\bar{o})}
                   {\exists_{\leq i} (q, o) & 
             \qquad \exists_{>j} (p,q)}    
\end{array}
\end{equation*}
are sound syllogistic rules in $\cS_z$ (and again in all larger languages).

Let $\sX$ be a set of syllogistic rules in $\cL$. A {\em substitution}
is a function $g: \bP \rightarrow \bP$; we extend $g$ to
$\cL$-formulas and to sets of $\cL$-formulas in the obvious way.  An
{\em instance} of a syllogistic rule $\Theta/\theta$ is the
syllogistic rule $g(\Theta)/g(\theta)$, where $g$ is a substitution.
Denote by $\bbP(\cL)$ the set of subsets of $\cL$.  We define the {\em
  direct syllogistic derivation relation} $\vdash_\sX$ to be the
smallest relation on $\bbP(\cL) \times \cL$ satisfying:
\begin{enumerate}
\item if $\theta \in \Theta$, then $\Theta \vdash_\sX \theta$;
\item if $\{\theta_1, \ldots, \theta_n\}/\theta$ is a syllogistic rule
  in $\sX$, $g$ a substitution, $\Theta = \Theta_1 \cup \cdots \cup
  \Theta_n$, and $\Theta_i \vdash_\sX g(\theta_i)$ for all $i$ ($1
  \leq i \leq n$), then $\Theta \vdash_\sX g(\theta)$.
\end{enumerate}
Where the language $\cL$ is clear from context, we omit
reference to it; further, we typically contract {\em syllogistic rule}
to {\em rule}.  Instances of the relation $\vdash_{\sX}$ can
always be established by {\em derivations} in the form of finite trees
in the usual way. For instance, the derivation
\begin{equation*}
\infer{\mZ(p,\bar{r})}
        {\leZ(o,r)                 &
         \infer{\mZ(p,o)}
               {\leZ(q,\bar{o}) &
                \mZ(p,q)}}
\end{equation*}
establishes that, for any set of syllogistic rules ${\sX}$
containing the rules~\eqref{eq:dariiFerio}, 
\begin{equation*}
\set{\leZ(o,r), \leZ(q,\bar{o}), \mZ(p,q)} \vdash_{\sX} \mZ(p,\bar{r}).
\end{equation*}
In the sequel, we reason freely about derivations in order to
establish properties of derivation relations. 

The derivation relation $\vdash_\sX$ is said to be {\em sound} if
$\Theta \vdash_\sX \theta$ implies $\Theta \models \theta$, and {\em
  complete} (for $\cL$) if $\Theta \models \theta$ implies $\Theta
\vdash_\sX \theta$. (Of course, this use of the word `complete' is
unrelated to the notion of a {\em complete} set of
$\cL(\bP')$-formulas defined in Sec.~\ref{sec:synsem}.) Evidently, it
is the existence of sound and complete derivation relations that
interest us, because they would yield convenient procedures for
discovering entailments such as~\eqref{eq:argS1}. A set $\Theta$ of
formulas is \emph{inconsistent} ({\em with respect to} $\vdash_\sX$)
if $\Theta \vdash_\sX \bot$ for some absurdity $\bot$; otherwise,
\emph{consistent}. It is obvious that, for any set of rules $\sX$,
$\vdash_\sX$ is sound if and only if every rule in $\sX$ is sound.

In the sequel, we will need to consider a stronger notion of
syllogistic derivation, incorporating a form of indirect reasoning.
Let $\cL$ be one of the languages considered above, and $\sX$ a set of
syllogistic rules in $\cL$. We define the {\em indirect syllogistic
  derivation relation} $\Vdash_\sX$ to be the smallest relation on
$\bbP(\cL) \times \cL$ satisfying:
\begin{enumerate}
\item if $\theta \in \Theta$, then $\Theta \Vdash_\sX \theta$;
\item if $\{\theta_1, \ldots, \theta_n\}/\theta$ is a syllogistic rule
  in $\sX$, $g$ a substitution, $\Theta = \Theta_1 \cup \cdots \cup
  \Theta_n$, and $\Theta_i \Vdash_\sX g(\theta_i)$ for all $i$ ($1
  \leq i \leq n$), then $\Theta \Vdash_\sX g(\theta)$.
\item if
  $\Theta \cup \set{\theta} \Vdash_\sX \bot$, where $\bot$ is any absurdity,
  then $\Theta \Vdash_\sX \bar{\theta}$.
\end{enumerate}
The only difference is the addition of the final clause, which allows
us to derive a formula $\bar{\theta}$ from premises $\Theta$ if we can
derive an absurdity from $\Theta$ together with $\theta$.  Instances
of the indirect derivation relation $\Vdash_{\sX}$ may also be
established by constructing derivations, except that we need a little
more machinery to keep track of premises. This may be done as
follows. Suppose we have a derivation (direct or indirect) showing
that $\Theta \cup \{\theta\} \Vdash_{\sX} \bot$, for some absurdity
$\bot$. Let this derivation be displayed as
\begin{equation*}
\infer*{\hspace{1mm} \bot,} 
       {\theta_1 & \cdots & \theta_n & \theta & \cdots & \theta}
\end{equation*}
where $\theta_1, \ldots, \theta_n$ is a list of formulas of $\Theta$
(not necessarily exhaustive, and with repeats allowed).  Applying
Clause~3 of the definition of $\Vdash_{\sX}$, we have $\Theta
\Vdash_{\sX} \bar{\theta}$, which we take to be established by the
derivation
\begin{equation*}
        \infer[\mbox{\small (RAA)}^1.]{\bar{\theta}} {\hspace{2mm}  \infer*{\bot}
          {\theta_1 &  \cdots &  \theta_n & [\theta]^1
            \cdots & [\theta]^1} \hspace{2mm}}
\end{equation*}
The tag (RAA) stands for {\em reductio ad absurdum}; the square
brackets indicate that the enclosed instances of $\theta$ have been
{\em discharged}, i.e.~no longer count among the premises; and the
numerical indexing is simply to make the derivation history clear.
Note that there is nothing to prevent $\theta$ from occurring among
the $\theta_1, \ldots, \theta_n$; that is to say, we do not have to
discharge all (or indeed any) instances of the premise $\theta$ if we
do not want to. 

The notions of {\em soundness} and {\em completeness} are defined for
indirect derivation relations in exactly the same way as for direct
derivation relations.  Again, it should be obvious that, for any set
of rules ${\sX}$, $\Vdash_\sX$ is sound if and only if every rule in
$\sX$ is sound. It is important to understand that {\em reductio ad
  absurdum} cannot be formulated as a syllogistic rule in the
technical sense defined here; rather, it is part of the
proof-theoretic machinery that converts any set of rules ${\sX}$ into
the derivation relation $\Vdash_{\sX}$. It is shown
in~\citet{s1:p-h+m09} that, for both the classical syllogistic,
$\cS_0$, and its extension $\cSd_0$, there exist finite sets of rules
$\sX$ such that the {\em direct} derivation relation $\vdash_\sX$ is
sound and complete. (That is: {\em reductio ad absurdum} is not
needed.)  However, the same paper considers various extensions of the
classical syllogistic for which there are sound and complete {\em
  indirect} syllogistic derivation relations, but no sound and
complete {\em direct} ones. Thus, it is in general important to
distinguish these two kinds of proof systems. We show below that
neither $\cS_z$ nor $\cSd_z$ has even a sound and complete {\em
  indirect} syllogistic proof system, for all $z >0$.

We end this section with two simple results on syllogistic derivations.
\begin{lemma}
Let $\cL$ be any of the languages $\cS_z$ or $\cSd_z$ \textup{(}$z
\geq 0$\textup{)}, $\sX$ a set of syllogistic rules in $\cL$, and
$\bP'$ a non-empty subset of $\bP$. Let $\theta \in \cL(\bP')$ and
$\Theta \subseteq \cL(\bP')$.  If there is a derivation
\textup{(}direct or indirect\textup{)} of $\theta$ from $\Theta$ using
$\sX$, then there is such a derivation involving only the atoms of
$\bP'$.  Further, if there is a derivation of an absurdity from
$\Theta$, then there is a derivation of an absurdity $\bot$ from
$\Theta$ such that $\bot \in \cL(\bP')$.
\label{lma:substitute}
\end{lemma}
\begin{proof}
Given a derivation of $\theta$ from $\Theta$, uniformly replace any
atom not in $\bP'$ with one that is. Similarly for the second
statement.
\end{proof}
\begin{lemma}
Let $\cL$ be any of the languages $\cS_z$ or $\cSd_z$ \textup{(}$z
\geq 0$\textup{)}, and $\sX$ a set of syllogistic rules in $\cL$.  Let
$\bP' \subseteq \bP$ be non-empty, and $\Psi$ a complete set of
$\cL(\bP')$-formulas.  If $\Psi \Vdash_\sX \bot$, then $\Psi
\vdash_\sX \bot$.
\label{lma:completeRAA}
\end{lemma}
\begin{proof}
Suppose that there is an indirect derivation of some absurdity $\bot$
from $\Psi$, using the rules $\sX$.  By Lemma~\ref{lma:substitute}, we
may assume all formulas involved are in $\cL(\bP')$.  Let the number
of applications of (RAA) employed in this derivation be $k$; and
assume without loss of generality that $\bot$ is chosen so that this
number $k$ is minimal. If $k > 0$, consider the last application of
(RAA) in this derivation, which derives a formula, say, $\bar{\psi}$,
discharging a premise $\psi$. Then there is an (indirect) derivation
of some absurdity $\bot'$ from $\Psi \cup \{\psi\}$, employing fewer
than $k$ applications of (RAA).  By minimality of $k$, $\psi \not \in
\Psi$, and so, by the completeness of $\Psi$, $\bar{\psi} \in
\Psi$. But then we can replace our original derivation of $\bar{\psi}$
with the trivial derivation, so obtaining a derivation of $\bot$ from
$\Psi$ with fewer than $k$ applications of (RAA), a
contradiction. Therefore, $k = 0$, or, in other words, $\Psi
\vdash_\sX \bot$.
\end{proof}

\section{No sound and complete syllogistic systems for $\cS_z$ or $\cSd_z$}
\label{sec:main}
In this section, we prove that none of the langauges $\cS_z$ or
$\cSd_z$ ($z >0$) has a sound and complete indirect syllogistic proof
system. The strategy we adopt is identical to that employed
in~\cite{s1:ph09} to obtain analogous results for the langauges
$\cN$ and $\cNd$.  However, the specific construction required to cope
with the restriction to $\cS_z$ and $\cSd_z$ is new, and more
involved.

To reduce clutter in the proof, we begin with the most
interesting case: $z = 1$.
\begin{theorem} There is no finite set $\sX$ of syllogistic rules in either
$\cS_1$ or $\cSd_1$ such that $\Vdash_{\mbox{\sX}}$ is sound and
  complete.
\label{theo:noSyll}
\end{theorem}
\begin{proof}
We first prove the result for $\cSd_1$; the result for $\cS_1$ will
then follow by an easy adaptation.  Henceforth, then, let $\sX$ be a
finite set of sound syllogistic rules in $\cSd_1$. We show that
$\Vdash_{\mbox{\sX}}$ is not complete.

Let $n \geq 4$, and let $\bP^n$ be a subset of $\bP$ of cardinality
$4n+2$: we shall write $\bP^n = \set{p_0, \dots, p_{2n-1}$, $q_0,
  \dots, q_{2n+1}}$. Now let $\Gamma^n$ be the set of
$\cSd_1(\bP^n)$-formulas given in~\eqref{pp++1}--\eqref{pq--1}, which,
for perspicuity, we divide into groups.  (Recall that the arguments of
formulas are taken to be {\em unordered}.) The first group concerns the
literals $p_i$ only:
\begin{align}
\label{pp++1} 
& \exO(p_i, p_{i+1}) & & (0 \leq i \leq 2n-2)\\
\label{pp++2} 
& \exO(p_i, p_{i+3}) & & (i \mbox{ even}) \wedge (0 \leq i \leq 2n-4)\\
\label{pp++3} 
& \leZ^*(p_{0}, p_{2n-1}) \\
\label{pp++4} 
& \mO^*(p_i, p_j) & & (0 \leq i \leq j \leq 2n-1) \wedge 
        (j \neq i+1) \wedge\\
\nonumber & & &
      (i \mbox{ odd} \vee j \neq i+3) \wedge
      (i \neq 0 \vee j \neq 2n-1)\\
\label{pp+-1} 
& \leZ^*(p_i, \bar{p}_i) & & (0 \leq i \leq 2n-1)\\
\label{pp+-2} 
& \mO^*(p_i, \bar{p}_j) & & (0 \leq i \leq 2n-1) \wedge (0 \leq j \leq 2n-1) \wedge
    (i \neq j)\\
\label{pp--1} 
& \mO^*(\bar{p}_i, \bar{p}_j) & & (0 \leq i \leq j \leq 2n-1).
\end{align}
The second group concerns the literals $q_i$ only:
\begin{align}
\label{qq++1} 
& \exO(q_{i}, q_{i+1}) & & (i \mbox{ even}) \wedge (0 \leq i \leq 2n)\\
\label{qq++2} 
& \mO^*(q_i, q_j) & & (0 \leq i \leq j \leq 2n+1) \wedge\\
\nonumber & & &
      (i \mbox{ odd} \vee j \neq i+1)\\
\label{qq+-1} 
& \leZ^*(q_i, \bar{q}_i) & & (0 \leq i \leq 2n+1)\\
\label{qq+-2} 
& \mO^*(q_i, \bar{q}_j) & & (0 \leq i \leq 2n+1) \wedge (0 \leq j \leq 2n+1) \wedge
    (i \neq j)\\
\label{qq--1} 
& \mO^*(\bar{q}_i, \bar{q}_j) & & (0 \leq i \leq j \leq 2n+1).
\end{align}
The third group mixes the literals $p_i$ and $q_i$:
\begin{align}
\label{pq++1} 
& \exO(p_{i+1}, q_{i}) & & (i \mbox{ even}) \wedge (0 \leq i \leq 2n-2)\\
\label{pq++2} 
& \exO(p_{i}, q_{i+1}) & & (0 \leq i \leq 2n-1)\\
\label{pq++3} 
& \exO(p_{i}, q_{i+3}) & & (i \mbox{ even} \wedge 0 \leq i \leq 2n-2)\\
\label{pq++4} 
& \mO^*(p_i, q_j) & & (0 \leq i \leq 2n-1) \wedge (0 \leq j \leq 2n+1) \wedge
   (j \neq i+1) \wedge\\
\nonumber & & &
      (i \mbox{ odd} \vee j \neq i+3) \wedge
      (j \mbox{ odd} \vee i \neq j+1)\\
\label{pq+-1} 
& \leZ^*(p_{i}, \bar{q}_{i}) & & (0 \leq i \leq 2n-1)\\
\label{pq+-2} 
& \leZ^*(p_{i}, \bar{q}_{i+2}) & & (0 \leq i \leq 2n-1)\\
\label{pq+-3} 
& \mO^*(p_i, \bar{q}_j) & & (0 \leq i \leq 2n-1) \wedge (0 \leq j \leq 2n+1) \wedge
   (j \neq i) \wedge\\
\nonumber & & &
   (j \neq i+2)\\
\label{pq-+1} 
& \mO^*(\bar{p}_i, q_j) & & (0 \leq i \leq 2n-1) \wedge (0 \leq j \leq 2n+1)\\
\label{pq--1} 
& \mO^*(\bar{p}_i, \bar{q}_j) & & (0 \leq i \leq 2n-1) \wedge (0 \leq j \leq 2n+1)
\end{align}
In fact, the formulas that will be doing most of the work here
are~\eqref{pp++1}, \eqref{pp++3}, \eqref{qq++1}, \eqref{pq+-1} and
\eqref{pq+-2}. The others are required only to ensure that we have a
complete set of formulas.
\begin{claim}
$\Gamma^n$ is a complete set of $\cSd_1(\bP^n)$-formulas. On the other
  hand, $\Gamma^n$ contains no absurdities.
\label{claim:complete}
\end{claim}
\begin{proof}
Consider first any formula of $\cSd_1(\bP^n)$ whose arguments are
$\set{p_i,p_j}$ ($0 \leq i \leq j \leq n-1$). Since the conditions on
$i$ and $j$ in~\eqref{pp++1}--\eqref{pp++4} are clearly exhaustive
(taking account of the fact that the arguments of formulas are
unordered), we see that $\Gamma^n$ contains one of $\leZ(p_i,p_j)$ or
$\mZ(p_i,p_j)$, and one of $\leO(p_i,p_j)$ or $\mO(p_i,p_j)$. The
other argument-patterns are dealt with similarly. The second part of
the lemma is ensured by the condition $(i \neq j)$ in the sets of
formulas~\eqref{pp+-2} and~\eqref{qq+-2}.
\end{proof}
\begin{claim}
$\Gamma^n$ is unsatisfiable.
\label{claim:unsatisfiable}
\end{claim}
\begin{proof}
Suppose $\fA \models \Gamma^n$. From~\eqref{pp++1}, let $a_h$ 
satisfy $p_{2h}$ and $p_{2h+1}$ ($0 \leq h \leq
n-1$). From~\eqref{pq+-1} and~\eqref{pq+-2}, $a_h$ also satisfies
$q_{2h}$, $q_{2h+1}$, $q_{2h+2}$ and $q_{2h+3}$. Hence,
from~\eqref{qq++1}, we have $a_0 = a_1 = \cdots = a_{n-1}$. But then this
common element satisfies $p_0$ and $p_{2n-1}$,
contradicting~\eqref{pp++3}.
\end{proof}

We now proceed to define a collection of satisfiable variants of $\Gamma^n_t$.
For all $t$ ($1 \leq t \leq n-2$), let
\begin{multline*}
\Gamma^n_t = (\Gamma^n \setminus 
              \set{\mZ(p_{2t-1},p_{2t}), 
                   \mZ(p_{2t-2},p_{2t+1}), 
                   \leO(q_{2t},q_{2t+1})}) \cup \\
             \set{\leZ(p_{2t-1},p_{2t}), 
                  \leZ(p_{2t-2},p_{2t+1}), 
                  \mO(q_{2t},q_{2t+1})}. 
\end{multline*}
Since $\Gamma^n$ is complete, so is $\Gamma^n_t$. The difference
between $\Gamma^n$ and $\Gamma^n_t$ that will be doing most of the
work here is that the latter set lacks the formula
$\leO(q_{2t},q_{2t+1})$. This disrupts the argument of
Claim~\ref{claim:unsatisfiable}: the best we can now infer is that
$a_0 = a_1 = \cdots = a_{t-1}$, and $a_t = a_{t+1} = \cdots =
a_{n-1}$, so that there need be no element satisfying both $p_0$ and
$p_{2n-1}$. 
\begin{claim}
If $1 \leq t < t' \leq n-2$, then $\Gamma^n_t \cap \Gamma^n_{t'}
\subseteq \Gamma^n$.
\label{claim:intersect}
\end{claim}
\begin{proof}
Straightforward check.
\end{proof}
To show that $\Gamma^n_t$ is satisfiable, we define a structure
$\fB^n_t$ as follows. The domain of $\fB^n_t$ consists of elements
$a$, $a'$, $b_{i,j}$, $b'_{i,j}$, $c_{i,j}$, $c'_{i,j}$, $d_{i,j}$,
$d'_{i,j}$, $e$ and $e'$, with indices subject to the conditions in
the middle column of Table~\ref{table:model} (interpreted
conjunctively); the atoms satisfied by these elements in $\fB^n_t$ are listed in
the right-most column of Table~\ref{table:model}. Note that the
elements $e$ and $e'$ satisfy no atoms at all. Roughly, the elements
$a$ and $a'$ ensure the truth of formulas in $\Gamma^n_t$ of the form
$\mZ(\ell, m)$ for which the corresponding $\leO(\ell, m)$ is {\em
  also} in $\Gamma^n_t$ (i.e.~a uniqueness claim), while the remaining
elements ensure the truth of formulas in $\Gamma^n_t$ of the form
$\mO(\ell, m)$.  The main task in the proof of
Claim~\ref{claim:satisfies} below is to ensure that these latter
elements do not spoil any uniqueness claims.
\begin{table}
\begin{center}
\begin{tabular}{|l|l|l|}
\hline
Element name & index conditions & atoms satisfied\\
$a$ & \ & $p_0$, \dots, $p_{2t-1}$, $q_0$, \dots, $q_{2t+1}$, \\
\hline
$a'$ & \ & $p_{2t}$, \dots, $p_{2n-1}$, $q_{2t}$, \dots, $q_{2n+1}$, \\
\hline
$b_{i,j}$, $b'_{i,j}$ & 
   $(0 \leq i \leq j \leq 2n-1) $ & 
   $p_i$, $q_i$, $q_{i+2}$, $p_j$, $q_j$, $q_{j+2}$\\
\ &  $(j \neq i+1)$ & \\
\ &  $(i \mbox{ odd} \vee j \neq i+3)$ & \\
\ & $(i \neq 0 \vee j \neq 2n-1) $ & \\
\hline
$c_{i,j}$, $c'_{i,j}$ & 
   $(0 \leq i \leq 2n-1)  $ & 
   $p_i$, $q_i$, $q_{i+2}$, $q_j$\\
\ & $(0 \leq j \leq 2n+1) $ & \\
\ & $(j \neq i+1) $ & \\
\ &  $(i \mbox{ odd} \vee j \neq i+3) $ & \\
\ &  $ (j \mbox{ odd} \vee i \neq j+1)$ & \\
\hline
$d_{i,j}$, $d'_{i,j}$ & 
   $(0 \leq i \leq j \leq 2n+1) $ & 
   $q_i$, $q_j$\\
\ &  $(i \mbox{ odd} \vee j \neq i+1)$ & \\
\hline
$e$, $e'$ & \ &\\
\hline
\end{tabular}
\end{center}
\caption{Definition of the model $\fB^n_t$: elements $b_{i,j}$,
  $b'_{i,j}$, $c_{i,j}$, $c'_{i,j}$, $d_{i,j}$ and $d'_{i,j}$ exist
  only for those pairs of indices $i, j$ satisfying {\em all} the
  indicated index conditions.}
\label{table:model}
\end{table}
\begin{claim}
For all $n \geq 4$ and all $t$ \textup{(}$1 \leq t \leq n-2$\textup{)}, 
$\fB^n_t \models \Gamma^n_t$.
\label{claim:satisfies}
\end{claim}
\begin{proof}
We consider the formulas~\eqref{pp++1}--\eqref{pq--1} in turn, taking
account of the differences between $\Gamma^n$ and $\Gamma^n_t$ as we
encounter them.

\begin{description}
\item[\eqref{pp++1}:] For $0 \leq i \leq 2t-2$, $a$ satisfies the atoms $p_i$
  and $p_{i+1}$, whereas $a'$ does not; for $2t \leq i \leq 2n-2$, $a'$ satisfies
  the atoms $p_i$ and $p_{i+1}$, whereas $a$ does not; for $i = 2t-1$,
  neither $a$ nor $a'$ satisfies (both) these atoms, but then $\Gamma^n_t$
  replaces $\mZ(p_{2t-1}, p_{2t})$ by $\leZ(p_{2t-1}, p_{2t})$. It remains to
  check that no other element satisfies these atoms. The only
  danger is from $b_{i',j'}$ and $b'_{i',j'}$, where $i'= i$ and $j' =
  i+1$; but the condition $j' \neq i' +1$ (middle column of
  Table~\ref{table:model}) rules this combination of  indices out.
\item[\eqref{pp++2}:] Almost identical to the argument for \eqref{pp++1}.
\item[\eqref{pp++3}:] Almost identical to the argument for \eqref{pp++1}.
\item[\eqref{pp++4}:] The elements $b_{i,j}$ and $b_{i',j'}$ both
  satisfy the atoms $p_i$ and $p_j$. Notice that the conditions on the
  indices in~\eqref{pp++4} are matched by the relevant conditions in the middle
  column of Table~\ref{table:model}, so that all formulas are
  accounted for.
\item[\eqref{pp+-1}:] Trivially satisfied.
\item[\eqref{pp+-2}:] If $i \neq j$, then both $c_{i,i}$ and
  $c'_{i,i}$, which exist for all $i$ ($0 \leq i \leq 2n-1$), satisfy
  the literals $p_i$, and $\bar{p}_j$.
\item[\eqref{pp--1}:] Both $e$ and $e'$ satisfy the literals
  $\bar{p}_i$ and $\bar{p}_j$.
\item[\eqref{qq++1}:] For $0 \leq i \leq 2t-2$, $a$ satisfies the
  atoms $q_{i}$ and $q_{i+1}$, whereas $a'$ does not; for $2t+2 \leq i
  \leq 2n$, $a'$ satisfies the atoms $q_{i}$ and $q_{i+1}$, whereas
  $a$ does not; for $i = 2t$, both $a$ and $a'$ satisfy (both) these
  atoms, but then $\Gamma^n_t$ replaces $\leO(q_{2t}, q_{2t+1})$ by
  $\mO(q_{2t}, q_{2t+1})$. It remains to check that no other element
  satisfies these atoms. Considering the right-hand column of
  Table~\ref{table:model}, the only danger is from: ({\em i})
  $b_{i',j'}$ and $b'_{i',j'}$, where $i'= i$ and $j' = i+1$; ({\em
    ii}) $b_{i',j'}$ and $b'_{i',j'}$, where $i'+2=i$ and $j' = i+1$
  (with $i'$ even); ({\em iii}) $b_{i',j'}$ and $b'_{i',j'}$, where
  $i'+2= i$ and $j'+2 = i+1$; ({\em iv}) $b_{i',j'}$ and $b'_{i',j'}$,
  where $i'+2 = i+1$ and $j' = i$; ({\em v}) $c_{i',j'}$ and
  $c'_{i',j'}$, where $i'= i$ and $j' = i+1$; ({\em vi}) $c_{i',j'}$
  and $c'_{i',j'}$, where $i'+2= i$ and $j' = i+1$ (with $i'$ even);
  ({\em vii}) $c_{i',j'}$ and $c'_{i',j'}$, where $j'= i$ and $i' =
  i+1$ (with $j'$ even); ({\em viii}) $c_{i',j'}$ and $c'_{i',j'}$,
  where $i'+2 = i+1$ and $j' = i$; ({\em ix}) $d_{i',j'}$ and
  $d'_{i',j'}$, where $i'= i$ and $j' = i+1$ (with $i'$ even).
  However, the conditions in the middle column of
  Table~\ref{table:model} rule these combinations of indices out.
\item[\eqref{qq++2}:] The elements $d_{i,j}$ and $d_{i',j'}$ both
  satisfy the atoms $q_i$ and $q_j$. Notice that the conditions on the
  indices in~\eqref{qq++2} are matched by the relevant conditions in the middle
  column of Table~\ref{table:model}, so that all formulas are accounted for.
\item[\eqref{qq+-1}:] Trivially satisfied.
\item[\eqref{qq+-2}:] If $i \neq j$, then both $d_{i,i}$ and
  $d'_{i,i}$, which exist for all $i$ ($0 \leq i \leq 2n+1$), satisfy
  the literals $q_i$, and $\bar{q}_j$.
\item[\eqref{qq--1}:] Both $o$ and $o'$ satisfy the literals
  $\bar{q}_i$ and $\bar{q}_j$.
\item[\eqref{pq++1}:] For $0 \leq i \leq 2t-2$, $a$ satisfies the
  atoms $p_{i+1}$ and $q_{i}$, whereas $a'$ does not; for $2t \leq i
  \leq 2n-2$, $a'$ satisfies the atoms $p_{i+1}$ and $q_{i}$, whereas
  $a$ does not. We check that no other element satisfies these
  atoms. Considering the right-hand column of Table~\ref{table:model},
  the only danger is from: ({\em i}) $b_{i',j'}$ and $b'_{i',j'}$,
  where $i'= i$ and $j' = i+1$; ({\em ii}) $b_{i',j'}$ and
  $b'_{i',j'}$, where $i'+2= i$ and $j' = i+1$ (with $i'$ even); ({\em
    iii}) $c_{i',j'}$ and $c'_{i',j'}$, where $i'= i+1$ and $j' = i$ (with
  $j'$ even).
  However, the conditions in the middle column of
  Table~\ref{table:model} rule these combinations of indices out.
\item[\eqref{pq++2}:] For $0 \leq i \leq 2t-1$, $a$ satisfies the
  atoms $p_i$ and $q_{i+1}$, whereas $a'$ does not; for $2t \leq i
  \leq 2n-1$, $a'$ satisfies the atoms $p_i$ and $q_{i+1}$, whereas
  $a$ does not. We check that no other element satisfies these
  atoms. Considering the right-hand column of Table~\ref{table:model},
  the only danger is from: ({\em i}) $b_{i',j'}$ and $b'_{i',j'}$,
  where $i'= i$ and $j' = i+1$; ({\em ii}) $b_{i',j'}$ and
  $b'_{i',j'}$, where $i'+2= i+1$ and $j' = i$; ({\em iii})
  $c_{i',j'}$ and $c'_{i',j'}$, where $i'= i$ and $j' = i+1$.
  However, the conditions in the middle column of
  Table~\ref{table:model} rule these combinations of indices out.
\item[\eqref{pq++3}:] For $0 \leq i \leq 2t-2$, $a$ satisfies the
  atoms $p_{i}$ and $q_{i+3}$, whereas $a'$ does not; for $2t \leq i
  \leq 2n-2$, $a'$ satisfies the atoms $p_{i}$ and $q_{i+3}$, whereas
  $a$ does not. We check that no other element satisfies these
  atoms. Considering the right-hand column of Table~\ref{table:model},
  the only danger is from: ({\em i}) $b_{i',j'}$ and $b'_{i',j'}$,
  where $i'= i$ and $j' = i+3$ (with $i'$ even); ({\em ii})
  $b_{i',j'}$ and $b'_{i',j'}$, where $i'= i$ and $j'+2 = i+3$; ({\em
    iii}) $c_{i',j'}$ and $c'_{i',j'}$, where $i'= i$ and $j' = i+3$
  (with $i'$ even).  However, the conditions in the middle column of
  Table~\ref{table:model} rule these combinations of indices out.
\item[\eqref{pq++4}:] The elements $c_{i,j}$ and $c_{i',j'}$ both
  satisfy the atoms $p_i$ and $q_j$. Notice that the conditions on the
  indices in~\eqref{pq++4} are matched by the relevant conditions in the middle
  column of Table~\ref{table:model}, so that all formulas are accounted for.
\item[\eqref{pq+-1}:] It is readily checked that every element satisfying $p_i$
  also satisfies $q_i$.
\item[\eqref{pq+-2}:] It is readily checked that every element satisfying $p_i$
  also satisfies $q_{i+2}$.
\item[\eqref{pq+-3}:] The elements $c_{i,i}$ and $c'_{i,i}$ both
  satisfy the literals $p_i$ and $\bar{q}_j$, as long as $j \neq i$ and 
  $j \neq i+2$. Note that these elements exist for all $i$ ($0 \leq i \leq 2n-1$),
  so all the relevant formulas are accounted for.
\item[\eqref{pq-+1}:] The elements $d_{j,j}$ and $d'_{j,j}$ both
  satisfy the literals $\bar{p}_i$ and $q_j$. Note that these elements
  exist for all $j$ ($0 \leq i \leq 2n+1$), so all the relevant
  formulas are accounted for.
\item[\eqref{qq--1}:] Both $e$ and $e'$ satisfy the literals
  $\bar{p}_i$ and $\bar{q}_j$.
\end{description}

\end{proof}

The key step in the proof is to show that, for any finite set of sound
syllogistic rules, we can make $n$ sufficiently large that those rules
cannot be used to infer anything new from $\Gamma^n$.  For suppose $\sX$
is a finite set of sound rules in $\cSd_1$. Let $r$ be the maximum
number of premises in any rule in $\sX$, and let $n \geq r+4$.
\begin{claim}
If $\Gamma^n \vdash_\sX \theta$, then $\theta \in \Gamma_n$.
\label{claim:stable}
\end{claim}
\begin{proof}
Consider any derivation establishing that $\Gamma^n \vdash_\sX
\theta$. By Lemma~\ref{lma:substitute}, we may assume that that
derivation features only atoms in $\bP^n$.  We show by induction on
the number of steps (proof-rule instances) in the derivation that
$\theta \in \Gamma^n$.  If there are no steps, then $\theta \in
\Gamma_n$ by definition. Otherwise, consider the last proof rule
instance, and let its antecedents be $\Theta$.  Thus, $|\Theta| \leq
n-4$; moreover, since the elements of $\Theta$ have shorter
derivations than $\theta$, $\Theta \subseteq \Gamma^n$, by inductive
hypothesis. Now consider the formulas
\begin{equation*}
  \mZ(p_{2h-1},p_{2h}), \quad \mZ(p_{2h-2},p_{2h+1}), \quad \leO(q_{2h},q_{2h+1}),
\end{equation*}
where $1 \leq h \leq n-2$; and let us arrange these formulas on a
rectangular grid, thus:
\begin{center}
\begin{tabular}{|c|c|c|c|}
$h=1$ & $h=2$ & \hspace{1cm} $\cdots$ \hspace{1cm} & $h=n-2$\\
\hline
$\mZ(p_1, p_2)$ & $\mZ(p_3, p_4)$ & $\cdots$ & $\mZ(p_{2n-5}, p_{2n-4})$\\
\hline
$\mZ(p_0, p_3)$ & $\mZ(p_2, p_5)$ & $\cdots$ & $\mZ(p_{2n-6}, p_{2n-3})$\\
\hline
$\leO(q_2, q_3)$ & $\leO(q_4, q_5)$ & $\cdots$ & $\leO(q_{2n-4}, q_{2n-3})$\\
\hline
\end{tabular}
\end{center}
Since $|\Theta| \leq n-4$, we can find two columns in this grid which
do not intersect $\Theta$: in other words, there exist two values of
$h$ ($1 \leq h \leq n-2$) such that $\Theta \subseteq \Gamma^n_h$. For
these values of $h$, $\fB^n_h \models \Theta$ by
Claim~\ref{claim:satisfies}, and hence $\fB^n_h \models \theta$, by
the supposed soundness of the rules in $\sX$. By the completeness of
$\Gamma^n_h$, $\theta \in \Gamma^n_h$, whence $\theta \in \Gamma^n$,
by Claim~\ref{claim:intersect}.
\end{proof}
We now complete the proof of the theorem.  Pick any absurdity
$\bot$. By Claim~\ref{claim:complete}, $\bot \not \in \Gamma^n$, and
so, by Claim~\ref{claim:stable}, $\Gamma^n \not \vdash_\sX \bot$. By
Lemma~\ref{lma:completeRAA}, $\Gamma^n \not \Vdash_\sX \bot$. Yet,
$\Gamma^n \models \bot$, by Claim~\ref{claim:unsatisfiable}. Thus,
$\Vdash_{\sX}$ is not complete, as required.

For $\cS_1$, simply delete from $\Gamma^n$ and $\Gamma^n_t$ all
formulas not in that language, and define $\fB^n_t$ as before. Since
the formulas~\eqref{pp++1}, \eqref{pp++3}, \eqref{qq++1},
\eqref{pq+-1} and \eqref{pq+-2} featuring in the proof of
Claim~\ref{claim:unsatisfiable} are all in $\cS_1$, and $\Gamma^n$ and
$\Gamma^n_t$ also differ only in respect of $\cS_1$-formulas, the
proof proceeds as for $\cSd_1$.
\end{proof}
\begin{corollary}
For all $z \geq 1$, there is no finite set $\sX$ of syllogistic rules
in either $\cS_z$ or $\cSd_z$ such that $\Vdash_{\mbox{\sX}}$ is sound
and complete.
\label{cor:noSyll}
\end{corollary}
\begin{proof}
For $z >1$, we make the following changes to the proof of
Theorem~\ref{theo:noSyll}. ({\em i}) In the definition of $\Gamma^n$,
all occurrences of $\mO^*$ are replaced by $\exists^*_{> z}$; and all
occurrences of $\exO$ are replaced by $\exists^*_{= 1}$.  ({\em ii})
The definition of $\Gamma^n_t$ in terms of $\Gamma^n$ is unaffected,
namely:
\begin{multline*}
\Gamma^n_t = (\Gamma^n \setminus 
              \set{\mZ(p_{2t-1},p_{2t}), 
                   \mZ(p_{2t-2},p_{2t+1}), 
                   \leO(q_{2t},q_{2t+1})}) \cup \\
             \set{\leZ(p_{2t-1},p_{2t}), 
                  \leZ(p_{2t-2},p_{2t+1}), 
                  \mO(q_{2t},q_{2t+1})}. 
\end{multline*}
Thus, $\Gamma^n_t$ contains (more precisely: {\em includes})
$\leZ^*(p_{2t-1},p_{2t})$, $\leZ^*(p_{2t-2},p_{2t+1})$ and
$\exists^*_{= 2}(q_{2t},q_{2t+1})$. ({\em iii}) In the definition of
$\fB^n_t$, instead of taking just two elements, $b_{i,j}$ and
$b'_{i,j}$, we instead take $(z+1)$ elements, $b_{i,j}, b'_{i,j}, \ldots,
b^{\prime \dots \prime}_{i,j}$; and similarly with $c_{i,j}$ and
$d_{i,j}$. The argument then proceeds exactly as for
Theorem~\ref{theo:noSyll}.
\end{proof}

\bibliographystyle{plainnat} 
\bibliography{s1}
\end{document}